\documentclass[12pt]{article}

\usepackage{hyperref}
\usepackage{comment}
\usepackage{amsmath} 
\usepackage{amssymb}
\usepackage{amsfonts}
\usepackage{amsthm}
\usepackage{graphicx}

\newtheorem{theorem}{Theorem}
\newtheorem{lemma}{Lemma}
\newtheorem{corollary}{Corollary}

\theoremstyle{definition}

\begin{document}
	\begin{center}
		\Large
	\textbf{Moments dynamics and stationary states for classical diffusion-type GKSL equations}
	
		\large 
		\textbf{D.D. Ivanov\footnote{Physics Faculty, Lomonosov Moscow State University, ul. Leninskyie gori 1, Moscow 119234, Russia \\ E-mail:\href{mailto:dan25032000@mail.ru}{dan25032000@mail.ru}}, A.E. Teretenkov}\footnote{Department of Mathematical Methods for Quantum Technologies, Steklov Mathematical Institute of Russian Academy of Sciences,
			ul. Gubkina 8, Moscow 119991, Russia\\ E-mail:\href{mailto:taemsu@mail.ru}{taemsu@mail.ru}}
		\end{center}
		
			\footnotesize
			The explicit dynamics of the moments for the GKSL equation and the approach in finding stationary Gaussian states are obtained. In our case the GKSL equation corresponds to Wiener stochastic processes. Such equations contain a double commutator which can be understood as a quantum analog of the second spatial derivative.

			
			\normalsize

	\section{Introduction}

	This paper is another approach of a series \cite{Ter19, Teretenkov20, NosTer20, Linowski2021} devoted to quantum kinetics equations in which dynamics of moments is exactly calculable. In this paper we consider the equation for the density matrix
	\begin{equation}\label{eq:mainEq}
	\frac{d}{dt} \rho_t = \mathcal{L}(\rho_t), \qquad \mathcal{L}(\rho) = -\frac12 \sum_{j=1}^{N} [C_j,[C_j, \rho]], 
    \end{equation}
    where $ C_j=C_j^{\dagger} $ --- operators, quadratic in bosonic or fermionic creation and annihilation operators. These generators arise in the case of averaging with respect to classical Wiener processes or, more generally, fields \cite{Kossakowski1972, Kummerer87, Holevo96, Holevo98}. Hence, since $ \mathcal{L} $ has the \textit{ Gorini-Kossakowski-Sudarshan-Lindblad} (GKSL) form \cite{gorini1976completely, lindblad1976generators}, we shall refer to it as a classical diffusion-type GKSL generator. 
        
    The main results of this paper are formulated in theorems 1 and 3 where for density matrix, satisfying $\eqref{eq:mainEq}$, the explicit dynamics of moments of creation and annihilation operators of arbitrary order for bosonic and fermionic cases respectively was obtained. Another group of results is stated in theorems 2 and 4 which are devoted to Gaussian stationary states of equation $\eqref{eq:mainEq}$.
        
    \section{Bosonic case}
    \ 
    
    We need to succinctly present some notation from \cite{Ter19, Ter16, Ter17a, Ter19R} to formulate our result. In this section we consider Hilbert space $\otimes_{j=1}^n\ell_2$.  Let us define the $2n$-dimensional vector of annihilation and creation operators $\mathfrak{a} = (\hat{a}_1, \cdots, \hat{a}_n, \hat{a}_1^{\dagger}, \cdots, \hat{a}_n^{\dagger} )^T$ satisfying canonical commutation relations \cite[Sec. 1.1.2]{scalli2003} $ [\hat{a}_i, \hat{a}_j^{\dagger}] = \delta_{ij}$, $ [\hat{a}_i, \hat{a}_j] = [\hat{a}_i^{\dagger}, \hat{a}_j^{\dagger}]= 0 $. One could denote linear and quadratic forms in such operators by $ f^T \mathfrak{a}  $ and $ \mathfrak{a}^T K \mathfrak{a} $, respectively. Here $  f \in \mathbb{C}^{2n} $ and $ K \in \mathbb{C}^{2n \times 2n} $. Let us define $2n \times 2n$-dimensional matrices as
    \begin{equation*}
	J = \biggl(
	\begin{array}{cc}
		0 & -I_n \\ 
		I_n & 0
	\end{array} 
	\biggr), \qquad
	E = \biggl(
	\begin{array}{cc}
		0 & I_n \\ 
		I_n & 0
	\end{array} 
	\biggr),
    \end{equation*}
    where $ I_n $ --- identity matrix from $ \mathbb{C}^{n \times n} $.  We also define the $\sim$-conjugation of vectors and matrices by the formulae
    \begin{equation*}
	\tilde{g} = E\overline{g}, \; g \in \mathbb{C}^{2n}, \qquad \tilde{K} = E \overline{K} E, \; K \in \mathbb{C}^{2n \times 2n},
    \end{equation*}
    where the overline is an (elementwise) complex conjugation.
    
    First of all, let us formulate the main results of this section and then we shall prove them. 
    \begin{theorem}\label{th:mainBos}
	Let the density matrix $\rho_{t}$ satisfy equation \eqref{eq:mainEq}, where  $ C_{j} = \frac{1}{2}\mathfrak{a}^{T}K_{j}\mathfrak{a}$, $K_{j} = K_{j}^{T} = \tilde{K_{j}} \in \mathbb{C}^{2n \times 2n} $ and $\langle\otimes_{l = 1}^{m}\mathfrak{a}\rangle_{0} < \infty$, then the dynamics of the moments has the form:
	\begin{equation}\label{eq:mometsDynBos}
		\langle\otimes_{l = 1}^{m}\mathfrak{a}\rangle_{t}=\exp \biggl( -\frac{t}{2}\sum\limits_{j=1}^N \sum\limits_{i,p = 1}^m\otimes_{l = 1}^m(JK_{j})^{\delta_{il}+\delta_{pl}} \biggr) \langle\otimes_{l = 1}^{m}\mathfrak{a}\rangle_{0},
	\end{equation}
	where $ \langle \; \cdot \; \rangle_t \equiv \operatorname{tr} (  \; \cdot \; \rho_t)$. In particular:
	\begin{align*}
		\langle\mathfrak{a}\rangle_{t} &=\exp \biggl(-\frac{t}{2}\sum\limits_{j=1}^N (JK_{j})^{2} \biggr)\langle\mathfrak{a}\rangle_{0},\\
		\langle\mathfrak{a}\otimes\mathfrak{a}\rangle_{t} &= \exp \biggl(-\frac{t}{2}\sum\limits_{j=1}^N(I_{2n}\otimes(JK_{j})^{2} + 2(JK_{j})\otimes(JK_{j}) + (JK_{j})^{2}\otimes I_{2n}) \biggr)\langle\mathfrak{a}\otimes\mathfrak{a}\rangle_{0}.
	\end{align*}
    \end{theorem}
    
    The following theorem provides us with the conditions for the Gaussian state having the form
    \begin{equation}\label{eq:GaussStBos}
	\rho^{\rm (st)} = \exp \biggl(\frac12 \mathfrak{a}^{T} M \mathfrak{a} + s \biggr)
    \end{equation}
    to be stationary solution of equation \eqref{eq:mainEq}. Let us stress that according to \cite[Sec.~3.3]{Ter19R} $ M=M^T = \tilde{M} $, $ ME < 0 $, $ e^{s} = \sqrt{|\det (e^{MJ} - I)|} $.
    \begin{theorem}\label{th:statStatesBos}
	A density matrix $\rho^{\rm (st)}$ is a stationary solution of equation \eqref{eq:mainEq}, where  $ C_{j} = \frac{1}{2}\mathfrak{a}^{T}K_{j}\mathfrak{a}$, $K_{j} = K_{j}^{T} = \tilde{K_{j}} \in \mathbb{C}^{2n \times 2n} $, having the form \eqref{eq:GaussStBos} if and only if
	\begin{equation}\label{eq:statCondBos}
		[K_jJ , e^{MJ}] = 0
	\end{equation}
	for all $ j = 1, \ldots, N $.
    \end{theorem}
    In order to prove theorem \ref{th:mainBos} we need to split it in several lemmas. To begin with, let us formulate a particular case of lemma 1 from \cite{Ter16} which will be utilized further.
    \begin{lemma}\label{lem:commBoson}
	Let $K = K^{T} \in \mathbb{C}^{2n \times 2n} $ and $ f \in \mathbb{C}^{2n} $, then
	\begin{equation*}
		\biggl[\frac{1}{2}\mathfrak{a}^{T}K\mathfrak{a},f^{T}\mathfrak{a} \biggr] = f^{T}(JK)\mathfrak{a}.
	\end{equation*}
    \end{lemma}
    
    We will utilize following multiplications of linear forms which will be arranged from left to right:
    \begin{equation*}
	\prod\limits_{l = 1}^m f_{l}^{T}\mathfrak{a} \equiv f_{1}^{T}\mathfrak{a} \ldots f_{m}^{T}\mathfrak{a},
    \end{equation*}
    where $ f_l \in \mathbb{C}^{2n} $ for all $ l = 1, \ldots, m $.
    \begin{lemma}\label{lem:commWithSevOpBoson}
	Let $K = K^{T} \in \mathbb{C}^{2n \times 2n} $ and $ f_l \in \mathbb{C}^{2n} $ for all $ l = 1, \ldots, m $, then
	\begin{equation}\label{eq:commSeveralQuadFormsBoson}
		 \left[\frac{1}{2}\mathfrak{a}^{T}K\mathfrak{a},\prod\limits_{l = 1}^mf_{l}^{T}\mathfrak{a}\right] = \sum\limits_{i = 1}^m\prod\limits_{l = 1}^mf_{l}^{T}(JK)^{\delta_{il}}\mathfrak{a}.
	\end{equation}
    \end{lemma}
    
    \begin{proof}
	It is convenient to prove this lemma by induction. The base case $ m=1 $ is right due to \ref{lem:commBoson}. Let \eqref{eq:commSeveralQuadFormsBoson} be right for $ m $, then let us prove that \eqref{eq:commSeveralQuadFormsBoson} is right for $ m+1 $. Evidently,
	\begin{align*}
		 \biggl[\frac{1}{2}\mathfrak{a}^{T}K\mathfrak{a},\prod\limits_{l = 1}^{m+1}f_{l}^{T}\mathfrak{a} \biggr] = \biggl[\frac{1}{2}\mathfrak{a}^{T}K\mathfrak{a},\prod\limits_{l = 1}^m f_{l}^{T}\mathfrak{a} \biggr]f_{m+1}^{T}\mathfrak{a} + \prod\limits_{l = 1}^m f_{l}^{T}\mathfrak{a} \biggl[\frac{1}{2}\mathfrak{a}^{T}K\mathfrak{a},f_{m+1}^{T}\mathfrak{a} \biggr] \\
		 = \sum\limits_{i = 1}^m \prod\limits_{l = 1}^m f_{l}^{T}  (JK)^{\delta_{il}}\mathfrak{a}f_{m+1}^{T}\mathfrak{a} + \prod\limits_{l = 1}^m f_{l}^{T}\mathfrak{a}f_{m+1}^{T}(JK)\mathfrak{a} = \sum\limits_{i = 1}^{m+1}\prod\limits_{l = 1}^{m+1}f_{l}^{T}(JK)^{\delta_{il}}\mathfrak{a}.
	\end{align*}
    \end{proof}

    \begin{corollary}
	\label{cor:doubleComm}
	Let $K = K^{T} \in \mathbb{C}^{2n \times 2n} $ and $ f_l \in \mathbb{C}^{2n} $ for all $ l = 1, \ldots, m $, then
	\begin{equation*}
		\left[\frac{1}{2}\mathfrak{a}^{T}K\mathfrak{a},\left[\frac{1}{2}\mathfrak{a}^{T}K\mathfrak{a},\prod\limits_{l = 1}^mf_{l}^{T}\mathfrak{a}\right] \right] = \sum\limits_{p = 1}^m\sum\limits_{i = 1}^m\prod\limits_{l = 1}^mf_{l}^{T}(JK)^{\delta_{il}+\delta_{pl}}\mathfrak{a}.
	\end{equation*}
    \end{corollary}
    
    \begin{proof}
	Applying lemma \ref{lem:commWithSevOpBoson} twice, we have
	\begin{equation*}
		\biggl[\frac{1}{2}\mathfrak{a}^{T}K\mathfrak{a},\biggl[\frac{1}{2}\mathfrak{a}^{T}K\mathfrak{a},\prod\limits_{l = 1}^mf_{l}^{T}\mathfrak{a}\biggr]\biggr] = \biggl[\frac{1}{2}\mathfrak{a}^{T}K\mathfrak{a}, \sum\limits_{i = 1}^m\prod\limits_{l = 1}^mf_{l}^{T}(JK)^{\delta_{il}}\mathfrak{a}\biggr] = \sum\limits_{p = 1}^m\sum\limits_{i = 1}^m\prod\limits_{l = 1}^m f_{l}^{T}(JK)^{\delta_{il}+\delta_{pl}}\mathfrak{a}.
	\end{equation*}
    \end{proof}
    
    \begin{lemma}
	\label{eq:conjGen}
	If $ \mathcal{L} $ is defined by formula \eqref{eq:mainEq}, $ \rho $ is a density matrix, $  X   $ is an operator such as $ \operatorname{tr} X  C_{j}^{2}\rho < \infty $, $ \operatorname{tr} X  \rho C_{j}^{2} < \infty $, $ \operatorname{tr} X  C_{j} \rho C_{j} < \infty $, then
	\begin{equation*}
		\operatorname{tr}  X \mathcal{L}(\rho) = \operatorname{tr} \mathcal{L}  (X)\rho.
	\end{equation*}
    \end{lemma}
    
    \begin{proof} 
    	From direct calculation we have
    	\begin{align*}
    		\operatorname{tr}  X \mathcal{L}(\rho) &=-\operatorname{tr} (X\frac{1}{2}\sum\limits_{j}[C_{j},[C_{j},\rho]]) = -\frac{1}{2}\sum\limits_{j}  \operatorname{tr} (X(C_{j}^{2}\rho - 2C_{j}\rho C_{j} + \rho C_{j}^{2}) \\
    		&= -\operatorname{tr} (\rho\frac{1}{2}\sum\limits_{j}[C_{j},[C_{j},X]]) = \operatorname{tr} \mathcal{L}  (X)\rho.
    	\end{align*}
    \end{proof}
    
    \begin{proof}[Proof of theorem \ref{th:mainBos}]
	Directly from lemma \ref{eq:conjGen} we have
	\begin{equation}\label{eq:averDynamics}
		\frac{d}{dt}\langle X \rangle_t = -  \frac{1}{2}\sum\limits_{j}\langle[C_{j},[C_{j},X ]])  \rangle_t.
	\end{equation}
	Let us take into account that $ C_{j} = \frac{1}{2}\mathfrak{a}^{T}K_{j}\mathfrak{a}$, where $K_{j} = K_{j}^{T} = \tilde{K_{j}} \in \mathbb{C}^{2n \times 2n} $, which provides $ C_{j}^{\dagger} = C_{j}  $, and let $ X = \prod\limits_{l = 1}^m f_{l}^{T}\mathfrak{a} $, then due to corollary \ref{cor:doubleComm} equation \eqref{eq:averDynamics} takes the form
	\begin{equation*}
		\frac{d}{dt}\langle  \prod\limits_{l = 1}^m f_{l}^{T}\mathfrak{a}  \rangle_t = -\frac12 \sum\limits_{p = 1}^m\sum\limits_{i = 1}^m\langle \prod\limits_{l = 1}^mf_{l}^{T}(JK)^{\delta_{il}+\delta_{pl}}\mathfrak{a} \rangle_t,
	\end{equation*}
	that due to arbitrariness of $ f_l $ is equivalent to
	\begin{equation*}
		 \frac{d}{dt} \langle\otimes_{l = 1}^{m}\mathfrak{a}\rangle_{t} = -\frac{1}{2}\sum\limits_{j}\sum\limits_{i,p = 1}^m(\otimes_{l = 1}^m(JK_{j})^{\delta_{il}+\delta_{pl}} )\langle\otimes_{l = 1}^{m}\mathfrak{a}\rangle_{t}.
	\end{equation*}
    Solution of this equation takes form \eqref{eq:mometsDynBos}.
    \end{proof}
    To prove theorem \ref{th:statStatesBos} we require a special case of lemma 2 from \cite{Ter16}, which can be formulated in the following way.
    
    \begin{lemma}\label{lem:expSand}
	Let $ K =  K^T , M = M^T \in \mathbb{C}^{2n \times 2n} $, then
	\begin{equation*}
		e^{\frac12 \mathfrak{a}^{T} M \mathfrak{a} } \frac12 \mathfrak{a}^{T} K \mathfrak{a}  e^{-\frac12 \mathfrak{a}^{T} M \mathfrak{a} } =  \frac12 \mathfrak{a}^{T} e^{- MJ} K e^{J M}\mathfrak{a}.
	\end{equation*}
    \end{lemma}

    \begin{proof}[Proof of theorem \ref{th:statStatesBos}]
	Stationary solution of equation \eqref{eq:mainEq} satisfies the equation
	\begin{equation*}
		\sum\limits_{j}[C_{j},[C_{j},\rho^{\rm (st)}]] = 0,
	\end{equation*}
	then we have
	\begin{equation*}
		0 = \operatorname{tr} \rho^{\rm (st)}\sum\limits_{j}[C_{j},[C_{j},\rho^{\rm (st)}]] =\operatorname{tr} \sum\limits_{j}[\rho^{\rm (st)},C_{j}][C_{j},\rho^{\rm (st)}]= \sum_j \operatorname{tr} Z_j^{\dagger} Z_j,
	\end{equation*}
	where $ Z_j = [C_{j},\rho^{\rm (st)}] $. Considering $ \operatorname{tr} Z_j^{\dagger} Z_j \geqslant 0 $ for arbitrary $ Z_j $, in our case we deduce that $ Z_j = 0 $, i.e. $ [C_{j},\rho^{\rm (st)}] = 0 $. This implies
	\begin{equation*}
	C_{j}= \rho^{\rm (st)}  C_{j} (\rho^{\rm (st)})^{-1}.
	\end{equation*}
    
    Considering \ref{lem:expSand} these conditions are equivalent to $ e^{- MJ} K_j e^{J M} = K_j $. And from $ e^{J M} J = J e^{M J} $ finally we have \eqref{eq:statCondBos}.
    \end{proof}
    
    \section{Fermionic case}
    \ 
    
    In the fermionic case we also need a bit of notation that originates from \cite{Ter19, Ter19R, Ter17}. There we consider the finite-dimensional Hilbert space $\mathbb{C}^{2^n}$. In such a space one could define \cite[p. 407]{Takht11} $n$ pairs of fermionic creation and annihilation operators satisfying canonical anticommutation relations: $ \{\hat{c}_i^{\dagger}, \hat{c}_j \} = \delta_{ij},  \{\hat{c}_i, \hat{c}_j\} = 0 $.  Let us define the $2n$-dimensional vector $\mathfrak{c} = (\hat{c}_1, \ldots, \hat{c}_n, \hat{c}_1^{\dagger}, \ldots, \hat{c}_n^{\dagger})^T$ of creation and annihilation operators. The quadratic forms in such operators we denote by $ \mathfrak{c}^T K \mathfrak{c} $, $ K \in \mathbb{C}^{2n \times 2n} $. Matrix $ E $ and $\sim$-conjugation are defined in the same way as they were in the previous section. 
    
    Since the proofs of fermionic analogs of theorems \ref{th:mainBos} and \ref{th:statStatesBos} are almost identical to those of the bosonic case, we will state only their formulations.
    \begin{theorem}\label{th:mainFer}
    Let density matrix $\rho_{t}$ satisfy equation \eqref{eq:mainEq}, where  $ C_{j} = \frac{1}{2}\mathfrak{c}^{T}K_{j}\mathfrak{c}$, $K_{j} = -K_{j}^{T} =- \tilde{K_{j}} \in \mathbb{C}^{2n \times 2n} $, then the dynamics of the moments takes the form:
 	\begin{equation}\label{eq:mometsDynFer}
 		\langle\otimes_{l = 1}^{m}\mathfrak{c}\rangle_{t}=\exp \biggl( -\frac{t}{2}\sum\limits_{j=1}^N \sum\limits_{i,p = 1}^m\otimes_{l = 1}^m(EK_{j})^{\delta_{il}+\delta_{pl}} \biggr) \langle\otimes_{l = 1}^{m}\mathfrak{c}\rangle_{0},
 	\end{equation}
 	where $ \langle \; \cdot \; \rangle_t \equiv \operatorname{tr} (  \; \cdot \; \rho_t)$. In particular:
 	\begin{align*}
 		\langle\mathfrak{c}\rangle_{t} &=\exp \biggl(-\frac{t}{2}\sum\limits_{j=1}^N (EK_{j})^{2} \biggr)\langle\mathfrak{c}\rangle_{0},\\
 		\langle\mathfrak{c}\otimes\mathfrak{c}\rangle_{t} &= \exp \biggl(-\frac{t}{2}\sum\limits_{j=1}^N(I_{2n}\otimes(E K_{j})^{2} + 2(E K_{j})\otimes(E K_{j}) + (E K_{j})^{2}\otimes I_{2n}) \biggr)\langle\mathfrak{c}\otimes\mathfrak{c}\rangle_{0}.
 	\end{align*}
    \end{theorem}
    
    In the following theorem we will consider Gaussian states similar to \eqref{eq:GaussStBos}, i.e. having the form
    \begin{equation}\label{eq:GaussStFer}
	\rho^{\rm (st)} = \exp \biggl(\frac12 \mathfrak{c}^{T} M \mathfrak{c} + s\biggr),
    \end{equation}
    where as in \cite[Sec.~4.3]{Ter19R} $ M = -M^T =-\tilde{M} $ and $ e^{-s} = \sqrt{\det(e^{KE} + I)} $.

    \begin{theorem}\label{th:statStatesFer}
	A density matrix $\rho^{\rm (st)}$ is a stationary solution of equation \eqref{eq:mainEq}, where $ C_{j} = \frac{1}{2}\mathfrak{c}^{T}K_{j}\mathfrak{c}$, $K_{j} = -K_{j}^{T} =- \tilde{K_{j}} \in \mathbb{C}^{2n \times 2n} $, of form \eqref{eq:GaussStFer} if and only if
	\begin{equation}\label{eq:statCondFer}
		[K_jE , e^{ME}] = 0
	\end{equation}
	for all $ j = 1, \ldots, N $.
    \end{theorem}

	\section{Conclusions}
	
	Let us stress that due to the Levy-Khintchine theorem for Levy fields \cite{Sato2001} general case of GKSL equations derived by averaging the unitary dynamics on Levy fields is completely described by the combination of the diffusion generator, studied in the present paper, the Poisson generator, investigated in \cite{Teretenkov20, NosTer20} and the deterministic push. Hence, the combination of these results provides us with the dynamics of moments for such a general case. 
	
	Let us note the recent success in the study of properties of GKSL equations with quadratic generators \cite{Agredo2021a, Agredo2021b, Barthel2021}. The fact that in this paper the exact dynamics and the stationary states for a higher order GKSL equation were derived makes it possible to conclude that many of these recent results can be applied to these equations as well.

\end{document}